\documentclass[a4paper,UKenglish]{lipics}
 
\usepackage{graphicx,cite}
\usepackage{color}

\usepackage{amsmath,amssymb}
\usepackage[ruled,vlined,linesnumbered]{algorithm2e}
\usepackage{pgf,tikz}
\usetikzlibrary{arrows}
\usepackage{xspace,units}
\usepackage{paralist}

\allowdisplaybreaks

\usepackage{microtype}

\theoremstyle{plain}
\newtheorem{proposition}{Proposition}
\theoremstyle{definition}
\newtheorem{exmpl}{Example}

\newcommand{\set}[2]{\{ #1 \; | \; #2 \}}

\newcommand{\vecp}{{\mathbf p}}

\newcommand{\vecx}{{\mathbf x}}

\newcommand{\ot}{\leftarrow}
\newcommand{\D}{\displaystyle}

\newcommand{\classNP}{{\sf NP}}

\newcommand{\classPPAD}{{\sf PPAD}}

\newcommand{\M}{\mathcal M}
\newcommand{\p}{\mbox{\boldmath $p$}}

\SetAlgoCaptionSeparator{.}

\let\oldnl\nl
\newcommand{\nonl}{\renewcommand{\nl}{\let\nl\oldnl}}

\title{Computing Equilibria in Markets with Budget-Additive Utilities}

\author[1]{Xiaohui Bei}
\author[2]{Jugal Garg}
\author[3]{Martin Hoefer}
\author[4]{Kurt Mehlhorn}
\affil[1]{Nanyang Technological University, 
  Singapore.\\ 
  \texttt{xhbei@ntu.edu.sg}}
\affil[2]{University of Illinois at Urbana-Champaign,
  USA.\\
  \texttt{jugal@illinois.edu}}
\affil[3]{MPI f\"ur Informatik and Saarland University,
  Germany. \\
  \texttt{mhoefer@mpi-inf.mpg.de}}
\affil[4]{MPI f\"ur Informatik,
  Germany. \\
  \texttt{mehlhorn@mpi-inf.mpg.de}}

\authorrunning{X. Bei, J. Garg, M. Hoefer, K. Mehlhorn} 

\Copyright{Xiaohui Bei, Jugal Garg, Martin Hoefer, Kurt Mehlhorn}


\serieslogo{}
\volumeinfo
  {Billy Editor and Bill Editors}
  {2}
  {Conference title on which this volume is based on}
  {1}
  {1}
  {1}
\EventShortName{}
\DOI{10.4230/LIPIcs.xxx.yyy.p}

\begin{document}

\maketitle

\begin{abstract}
We present the first analysis of Fisher markets with buyers that have budget-additive utility functions. Budget-additive utilities are elementary concave functions with numerous applications in online adword markets and revenue optimization problems. They extend the standard case of linear utilities and have been studied in a variety of other market models. In contrast to the frequently studied CES utilities, they have a global satiation point which can imply multiple market equilibria with quite different characteristics. Our main result is an efficient combinatorial algorithm to compute a market equilibrium with a Pareto-optimal allocation of goods. It relies on a new descending-price approach and, as a special case, also implies a novel combinatorial algorithm for computing a market equilibrium in linear Fisher markets. We complement this positive result with a number of hardness results for related computational questions. We prove that it is \classNP-hard to compute a market equilibrium that maximizes social welfare, and it is \classPPAD-hard to find any market equilibrium with utility functions with separate satiation points for each buyer and each good.
\end{abstract}


\section{Introduction}
The concept of market equilibrium is a fundamental and well-established notion in economics to analyze and predict the outcomes of strategic interaction in large markets. Initiated by Walras in 1874, the study of market equilibrium has become a cornerstone of microeconomic analysis, mostly due to general results that established existence under very mild conditions~\cite{ArrowD54}. Since efficient computation is a fundamental criterion to evaluate the plausibility of equilibrium concepts, the algorithmic aspects of market equilibrium are one of the central domains in algorithmic game theory. Over the last decade, several new algorithmic approaches to compute market equilibria were discovered. Efficient algorithms based on convex programming
techniques can compute equilibria in a large variety of domains~\cite{CodenottiRChapter07,GoelV11,Jain07}. More importantly, several approaches were
proposed that avoid the use of heavy algorithmic machinery and follow combinatorial strategies~\cite{DevanurPSV08,JainV10,Orlin10,Vegh14STOC,DuanM15,DuanGM16}, or even work as a t\^atonnement process in unknown market environments~\cite{ColeF08,CheungCD13,BeiGH15}. Designing such combinatorial algorithms is useful also beyond the study of markets, since the underlying ideas can be applied in other areas. Variants of these algorithms were shown to solve scheduling~\cite{ImKM14,ImKM15} and cloud computing problems~\cite{DevanurGMVY15}, or can be used for fair allocation of indivisible items~\cite{ColeG15}.

In this paper, we design a new combinatorial polynomial time algorithm for computing equilibria in Fisher markets with budget-additive utilities. In a Fisher market, there is a single seller with a set $G = \{1,\ldots,m\}$ of goods. W.l.o.g.\ we assume that the total quantity of each good is 1. There is a set $B = \{1,\ldots,n\}$ of buyers. Each buyer $i$ has a budget $M_i > 0$ of money and a utility function $u_i$. For budget-additive utilities, $u_{ij} \ge 0$ is the utility of buyer $i$ if one unit of good $j$ is allocated to her. There is a \emph{happiness cap} $c_i > 0$, and the utility function is 
\[ u_i(\vecx_i) = \min \left\{ c_i, \sum_{j \in G} u_{ij} x_{ij} \right\} \enspace, \]
where $\vecx_i = (x_{ij})_{j \in G}$ is any bundle of goods assigned to buyer $i$. If $u_i(\vecx_i) = c_i$, then buyer $i$ is called \emph{capped buyer} for allocation $\vecx$. We assume all $u_{ij}$, $c_i$, $M_i$ are rational numbers.

Our goal is to compute an \emph{allocation} $\vecx = (\vecx_i)_{i\in B} $ of goods and \emph{prices} $\vecp = (p_j)_{j \in G}$ such that the pair $(\vecx,\vecp)$ is a market equilibrium. Given prices $\vecp$, a \emph{demand bundle} $\vecx_i^*$ of buyer $i$ is a bundle of goods that maximizes the utility
of buyer $i$ for its budget, i.e., $\vecx_i^* \in \arg\max \left\{ u_i(\vecx_i) \mid \sum_{j} p_j x_{ij} \le M_i \right\}$.
Note that $\sum_j u_{ij} x^*_{ij} > c_i$ is allowed.
A \emph{market equilibrium} $(\vecx,\vecp)$ is a pair such that 
\begin{itemize}
\item $\vecp \ge 0$ (prices are nonnegative),
\item $\sum_i x_{ij} \le 1$ for every $j \in G$ (no overallocation),
\item $\vecx_i$ is a demand bundle for every $i \in B$, and 
\item Walras' law holds: $p_j (1 - \sum_i x_{ij}) = 0$ for every $j \in G$. 
\end{itemize}

Note that if $\sum_i x_{ij} < 1$, then $p_j = 0$. An equilibrium $(\vecx,\vecp)$ is \emph{Pareto-optimal} if there is no equilibrium $(\vecx',\vecp')$ such that $u_i(\vecx) \le u_i(\vecx')$ for all $i$ and $u_i(\vecx) < u_i(\vecx')$ for at least one $i$.

Budget-additive utility functions are a simple class of submodular and concave functions and a natural generalization of the standard and well-understood case of linear utilities. These utility functions arise naturally in cases where agents have an intrinsic upper bound on their utility. For example, if the goods are food and the utility of a food item for a particular buyer is its calorie content, calories above a certain threshold do not increase the utility of the buyer. 
In addition, there are a variety of further applications in adword auctions and revenue maximization problems~\cite{AndelmanM04,AzarBKMN08,ChakrabartyG10,BuchbinderJN07}. Recently, market models where agents have budget-additive utilities attracted a significant amount of research interest, e.g., for the allocation of indivisible goods in offline~\cite{AndelmanM04,AzarBKMN08,ChakrabartyG10} and online~\cite{BuchbinderJN07,KapralovPV13} scenarios, for truthful mechanism design~\cite{BuchfuhrerDFKMPSSU10}, and for the study of Walrasian equilibrium with quasi-linear utilities~\cite{FeldmanGL13,DobzinskiFTW15,RoughgardenT15}. As simple variants of submodular functions, they capture many of the inherent difficulties of more general domains. Given this amount of interest, it is perhaps surprising that they are not well-understood within the classic Fisher and exchange markets. 
\smallskip

{\noindent \bf Results and Contribution. $\;$}
We study Fisher markets with budget-additive utilities. Our initial observations about these markets reveal that they have different properties than the ones with CES utilities usually studied in the literature. Due to the satiated nature of the utilities, capped buyers might not spend all their money or spend money on goods that do not give them maximum utility per unit of money, so prices and utilities in market equilibrium are not unique and can be quite different. It is possible to simply ignore the satiation and assume linear utilities. Then a variety of existing algorithms~\cite{DevanurPSV08,Orlin10,DuanM15,DuanGM16,BeiGH15} can be used to compute a market equilibrium. It continues to be a market equilibrium for the market with budget-additive utilities. However, this equilibrium may be undesirable, as in many cases it does not even satisfy Pareto-optimality of the allocation. 

\begin{exmpl}\label{ex:linear}
Consider a \emph{linear} market with two buyers and two goods, $u_{11} = 5$, $u_{21} = 2$, and $u_{12} = u_{22} = 1$. The budgets are $M_1 = 3$ and $M_2 = 1$. For the unique equilibrium we allocate good 1 completely to buyer 1 and good 2 completely to buyer 2, i.e., $x_{11} = x_{22} = 1$. The buyers' utilities amount to 5 and 1, resp., and the prices are $p_1 = 3$ and $p_2 = 1$. 
 
Now suppose buyer 1 has a budget-additive utility function with cap $c_1 = 1$. Then $(\vecx, \vecp)$ described above remains an equilibrium, since both buyers obtain a demand bundle (buyer 1 now has utility 1 instead of 5). Alternatively, suppose we allocate good 1 completely to buyer 2 and good 2 completely to buyer 1, i.e., $x_{12} = x_{21} = 1$. The utilities amount to 1 and 2, resp., and the prices can be chosen as $p_1 = 1$ and $p_2 \in [0.5,3]$. Here buyer 1 buys a bundle of goods with optimal utility of 1. Buyer 2 buys a demand bundle since he spends all its budget on a good that gives him the maximum {\it bang-per-buck ratio}. All goods are exactly allocated, and Walras' law holds. Thus, it represents another market equilibrium.
Note if $p_2 < 3$, buyer 1 does not spend all of its money, but it is still a demand bundle for because he achieves the maximum utility. Furthermore, such an equilibrium Pareto-dominates the one derived from the linear case in terms of utilities. \hfill $\blacksquare$
\end{exmpl}

We strive to compute a market equilibrium with a Pareto-optimal allocation and focus on a subset of market equilibria, in which we restrict the allocation to demand bundles which we call thrifty and modest -- buyers spend the least amount of money that can achieve their optimal utilities and receive a bundle of goods that has a minimality property. In Section~\ref{sec:prelim}, we show that such \emph{modest MBB equilibria} can be captured by a generalization of classic Eisenberg-Gale convex program, and with this additional property the utilities are unique and the allocation is always Pareto-optimal (w.r.t.\ all possible allocations, attainable in market equilibrium or not). We highlight that the set of modest MBB equilibria can be partially ordered with respect to their price vectors and forms a lattice. As such, there are modest MBB equilibria with pointwise largest and smallest prices, resp. Among all modest MBB equilibria they yield maximum and minimum revenue for the seller, resp.

Section~\ref{sec:algorithm} contains our main contribution -- a combinatorial algorithm that computes price and allocation vectors of a modest MBB equilibrium in time $O(mn^6(\log(m+n) + (m+n)\log U))$, where $n$ is the number of agents, $m$ the number of goods, and $U$ the largest integer in the market parameters. The computed equilibrium has a Pareto-optimal allocation, as well as pointwise largest prices and maximum revenue among all modest MBB equilibria.

Our algorithm represents a novel approach to compute market equilibria based on the idea of descending prices. While some parts of our algorithm are in spirit of combinatorial algorithms for linear markets~\cite{DevanurPSV08,DuanM15,DuanGM16,BeiGH15}, all these approaches are ascending-price algorithms. This technique and its usual analysis based the $1$-norm of excess money does not apply in our case, since the norm is non-monotonic and cannot be used to measure progress towards equilibrium. Surprisingly, our novel descending-price approach overcomes the $1$-norm issue, but we need to address additional challenges in establishing polynomial running time due to varying and non-increasing active budgets, and in showing that intermediate prices remain polynomially bounded. Note that, as a special case, this also yields a new combinatorial descending-price algorithm for linear Fisher markets.

In Section~\ref{sec:minPrice} we exploit the lattice structure of modest MBB equilibria and design a procedure, using which we can turn any modest MBB equilibrium into one with smallest prices and minimum revenue. In combination with the descending-price algorithm, it computes a modest MBB equilibrium with minimum revenue within the same asymptotic time bound.


Finally, we study two extensions in Section~\ref{sec:extend}. Facing multiple equilibria, a natural goal is to compute an allocation that maximizes
utilitarian social welfare. We prove that this problem is \classNP-hard, even when social welfare is measured by a $k$-norm of the vector of buyer
utilities, for any constant $k > 0$. Moreover, we consider a variant of the budget-additive utilities with a satiation point for \emph{each buyer and
each good}. They constitute a special class of separable piecewise-linear concave (SPLC) utilities, where each piecewise-linear component consists of
two segments with the second one being constant. We show that even in this very special case computing any market equilibrium becomes \classPPAD-hard.
\medskip

\noindent{\bf Related Work. }
The computation of market equilibria is a central area in algorithmic game theory. There are a variety of polynomial-time algorithms to compute approximate market equilibria based on solving different convex programming formulations~\cite{CodenottiRChapter07,GoelV11,Jain07}. Our paper is closer to work on markets with linear utilities and combinatorial algorithms that compute an exact equilibrium in polynomial time~\cite{DevanurPSV08,Orlin10,DuanM15,DuanGM16}. Directly related to our approach is the classic combinatorial algorithm for linear Fisher markets~\cite{DevanurPSV08}. In contrast, our algorithm is based on a new descending price approach where buyers are always saturated and goods have non-negative surplus. Further, the \emph{active} budgets of buyers vary with the price change, which creates new challenges in establishing a polynomial bound on the  number of iterations and the representation size of intermediate prices.

Independently of our work, Devanur et al~\cite{DevanurJMVY16} very recently presented the same convex program for Fisher markets with satiated buyers. They propose a polynomial-time algorithm for finding an arbitrary modest MBB equilibrium, but it is based on the ellipsoid method without any explicit running time bound.

Recently, algorithmic work has also started to address unknown markets, where utilities and budgets of buyers are unknown. Instead, algorithms iteratively set prices and query a demand oracle. In this domain, t\^atonnement dynamics have been studied for Fisher markets and extensions with concave utilities. For many classes of these markets, a notion of $(1+\epsilon)$-approximate market equilibrium can be reached after a convergence time polynomial in $1/\epsilon$ and other market parameters~\cite{ColeF08,CheungCD13,CheungCR12, BirnbaumDX11}. In some cases, the convergence time can even be reduced to $\log(1/\epsilon)$~\cite{CheungCD13}. A similar convergence rate is obtained by a more general algorithm even for general unknown exchange markets with weak gross-substitutes property, and even for linear markets with non-continuous demands and oracles using suitable tie-breaking~\cite{BeiGH15}. 

Allocation of indivisible items to agents with budget-additive utilities is an active area of research interest. There are constant-factor approximation algorithms for optimizing the allocation in offline~\cite{AndelmanM04,AzarBKMN08,ChakrabartyG10} and online~\cite{BuchbinderJN07,KapralovPV13} scenarios. Closer to our work is the study of markets with money. The existence of Walrasian equilibrium with quasi-linear utilities and algorithmic issues of bundling items were studied in~\cite{FeldmanGL13,DobzinskiFTW15,RoughgardenT15}. Strategic agents and truthful mechanisms for budget-additive markets have been analyzed in~\cite{BuchfuhrerDFKMPSSU10}. There are strong lower bounds for the approximation ratio of certain classes of truthful mechanisms, and a truthful mechanism with constant-factor approximation for budget-additive utilities is one of the most interesting open problems in combinatorial auctions.


\section{Preliminaries}
\label{sec:prelim}

For a given price vector $\vecp$ and buyer $i$, we denote the \emph{maximum bang-per-buck (MBB)} ratio by $\alpha_i = \max_j u_{ij}/p_j$, where we make the assumption that $0/0 = 0$. Budget-additive utilities strictly generalize linear utilities: when all $c_i$'s are large enough, they are equivalent to linear utilities. If buyer $i$ is uncapped in a market equilibrium $(\vecx,\vecp)$, it behaves as in the linear case, spends all its budget, and buys only MBB goods ($x_{ij} > 0$ only if $u_{ij}/p_j = \alpha_i$). Otherwise, 
if buyer $i$ is capped in $(\vecx,\vecp)$, it might buy non-MBB goods and not spend all of its budget. This implies that unlike the case of linear utilities, market equilibrium prices and utilities are not unique with budget-additive utilities.

It is easy to see that we can obtain one market equilibrium by simply ignoring the happiness caps and treating the market as a linear one. However, this equilibrium is often undesirable since it is not always Pareto-optimal. 

Our main goal in this paper is to find a market equilibrium that is Pareto-optimal. More generally, we will also be concerned with finding a
(Pareto-optimal) market equilibrium that can maximize social welfare $\sum_{i \in B} u_{i} (\vecx_{i})$. For the former we provide a polynomial-time
algorithm, the latter we prove it to be \classNP-hard.

\medskip
\smallskip

{\noindent \bf Modest MBB Equilibria, Pareto-Optimality, and Uniqueness. $\;$} 
The main challenges in budget-additive markets arise from capped buyers, who may possibly have multiple choices for the demand bundle. Let us introduce two convenient restrictions on the allocation to capped buyers. 

\begin{itemize}
\item An allocation $\vecx_i$ for buyer $i$ is called \emph{modest} if $\sum_j u_{ij}x_{ij} \le c_i$. By definition, for uncapped buyers every demand bundle is modest. For capped buyers, a modest bundle of goods $\vecx_i$ is such that utility breaks even between the linear part and $c_i$, i.e., $c_i = \sum_j u_{ij}x_{ij}$. 

\item A demand bundle $\vecx_i$ is called \emph{thrifty} or \emph{MBB} if it consists of only MBB goods: $x_{ij} > 0$ only if $u_{ij}/p_j = \alpha_i$. As noted above, for uncapped buyers every demand bundle is MBB. 
\end{itemize}

We call a market equilibrium $(\vecx,\vecp)$ a \emph{modest MBB equilibrium} if $\vecx_i$ is modest and MBB for every buyer $i \in B$. We show an algorithm to compute in polynomial time such an equilibrium where $\vecx$ is also Pareto-optimal. Such an equilibrium is also desirable because it agrees the behavioral assumption that each buyer is thrifty and spends the least amount of money in order to obtain a utility maximizing bundle of goods.
 
Consider the following Eisenberg-Gale program~\eqref{pro:EG}, which allows us to find a modest and Pareto-optimal allocation.
\begin{equation}\label{pro:EG}
	\begin{array}[4]{rrcll}
	\text{Max.} & \multicolumn{4}{l}{\D \sum_{i \in B} M_i \log \D \sum_{j \in G} u_{ij} x_{ij} }\vspace{0.2cm}\\
	\text{s.t.}	& \D \sum_{j \in G} u_{ij} x_{ij}& \leq & c_i & i \in B \vspace{0.2cm}\\
	          	& \D \sum_{i \in B} x_{ij} & \leq & 1 & j \in G\\
						  & x_{ij} & \ge & 0 & i \in B,\ j \in G\\
	\end{array}
	\end{equation}
	By standard arguments, we consider the dual for~\eqref{pro:EG} using dual variables $\gamma_i$ and $p_j$ for the first two constraints, resp., and the KKT conditions read: 
	\medskip

\begin{minipage}[b]{0.49\linewidth}
	\begin{enumerate}
	\item \label{BPP1} $p_j/u_{ij} \ge M_i/u_i - \gamma_i$
	\item \label{BPP2} $x_{ij} > 0 \; \Rightarrow \; p_j/u_{ij} = M_i/u_i - \gamma_i$
	\end{enumerate}
\end{minipage}	
\begin{minipage}[b]{0.49\linewidth}
	\begin{enumerate}
	\setcounter{enumi}{2}	
	\item $p_j \ge 0$ and $p_j > 0 \; \Rightarrow \; \sum_{i \in B} x_{ij} = 1$
	\item $\gamma_i \ge 0$ and $\gamma_i > 0 \; \Rightarrow \; u_i = c_i$
	\end{enumerate}
\end{minipage}	
\vspace{0.01cm}

Observe that the Lagrange multiplier $\gamma_i$ indicates if the cap $c_i$ represents a tight constraint in the optimum solution. The dual variables $p_j$ can be interpreted as prices. Note that conditions~\ref{BPP1} and~\ref{BPP2} imply that $x_{ij} > 0$ if and only if 
  $j \in \arg \min_{j'} p_{j'}/u_{ij'} = \arg \max_{j'} u_{ij'}/p_{j'}$,
i.e., all agents purchase goods with maximum bang-per-buck. Hence, similarly as for linear markets~\cite{VaziraniChapter07}, the KKT conditions imply that an optimal solution to the EG program~\eqref{pro:EG} and corresponding dual prices constitute a market equilibrium, in which every agent buys goods that have maximum bang-per-buck. The KKT conditions postulate this also for agents whose utility reaches the cap. Thus, the optimal solution to this program is a modest MBB equilibrium. Furthermore, we obtain the following favorable analytical properties.
\begin{proposition} \label{prop:EGisMBB}
The optimal solutions to~\eqref{pro:EG} are exactly the modest MBB equilibria. The utility vector is unique across all such equilibria and each such equilibrium is Pareto-optimal. In particular, there is a unique set of capped buyers. Non-capped buyers spend all their money. Capped buyers do not overspend.
\end{proposition}
\begin{proof} 
We observe first that there is an interior feasible solution to \eqref{pro:EG}. Simply set $x_{ij} = \epsilon > 0$ for all $i$ and $j$, where $\epsilon$ is small enough such all constraints in~\eqref{pro:EG} are satisfied with inequality. The existence of an interior feasible solution guarantees that the KKT conditions are necessary and sufficient for an optimal solution to~\eqref{pro:EG}. 

Let $\vecx$ and $\vecx'$ be two optimal solutions to (\ref{pro:EG}) and assume that $u_h(\vecx) \not= u_h(\vecx')$ for some buyer $h$. Consider the allocation $\vecx'' = (\vecx + \vecx')/2$. It is clearly feasible. Also, 
\[ \sum_{i \in B} M_i \log u_i (\vecx'') < \left( \sum_{i \in B} M_i \log u_i(\vecx) + \sum_{i \in B} M_i \log u_i(\vecx') \right)/2\enspace, \]
a contradiction to the optimality of the allocation. The inequality follows from the concavity of the $\log$-function. We have now shown that the utilities of the buyers are unique among all optimal solutions of \eqref{pro:EG}. Thus, every optimal solution to \eqref{pro:EG} is modest, MBB and Pareto-optimal.

Conversely, let $(\vecx,\vecp)$ be a modest MBB equilibrium. We show that $\vecx$ is an optimal solution to (\ref{pro:EG}). $\vecx$ is feasible since it is modest and does not overallocate any good. Since $\vecx_i$ is a thrifty demand bundle for buyer $i$, we have $u_{ij}/p_j = \alpha_i = \max_\ell u_{i\ell}/p_\ell$ whenever $x_{ij} > 0$. Thus
\[  M_i \ge \sum_{j} p_j x_{ij} = \sum_j \frac{u_{ij}}{\alpha_i} x_{ij} = \frac{u_i(\vecx)}{\alpha_i}\enspace, \]
and hence $M_i/u_i(\vecx) \ge 1/\alpha_i$. Let $\gamma_i = M_i/u_i - 1/\alpha_i$. Then $\gamma_i \ge 0$. We show that the KKT conditions hold. For any $j$, we have $p_j/u_{ij} \ge 1/\alpha_i = M_i/u_i - \gamma_i$. If $x_{ij} > 0$, then $p_j/u_{ij} = 1/\alpha_i = M_i/u_i - \gamma_i$. Prices are non-negative by definition and $p_j > 0$ implies that good $j$ is completely allocated by Walras's law. Finally, assume $\gamma_i > 0$. Then $M_i/u_i(\vecx) > 1/\alpha_i$ and hence 
\[ M_i > \frac{u_i(\vecx)}{\alpha_i} = \frac{\sum_j u_{ij} x_{ij}}{\alpha_i} = \sum_j p_j x_{ij}\enspace, \]
where the first equality follows from the fact that the allocation is modest. Let $\delta = M_i/\sum_j p_j x_{ij}$. Then buyer $i$ could afford the bundle $\delta \vecx_i$. Since $\vecx_i$ is a demand bundle for buyer $i$, we must have $c_i \le \sum_j u_{ij} x_{ij}$. Since the allocation is modest, we have equality. 
\end{proof}
While utilities are unique, allocation and prices of modest MBB equilibria might not be unique. Consider a market with two identical buyers and two goods, where $u_{11} = u_{12} = u_{21} = u_{22} = 1$, $c_1 = c_2 = 1$, and $M_1 = M_2 = 5$. The unique equilibrium utility of both buyers is $u_1 = u_2 = 1$, which can be obtained for any $p_1 = p_2 = p$, where $p \in [0,5]$ and allocation $\vecx$ satisfying $x_{11} + x_{12} = 1;\ x_{21} + x_{22} = 1;\ x_{11} + x_{21} = 1;\ x_{12} + x_{22} = 1$.


\setcounter{exmpl}{0}
\begin{exmpl}[continued] 
For our example above, the modest MBB equilibrium obtained from solving the convex program is $x_{11} = 1/5$, $x_{12} = 0$, $x_{21} = 4/5$ and $x_{22} = 1$ with prices $p_1 = 10/13$ and $p_2 = 5/13$. Buyer 1 spends $2/13$, buyer 2 spends the entire budget. The utilities are $1$ and $13/5$. It is easy to see that the KKT conditions hold. This equilibrium is Pareto-optimal and also the best equilibrium in terms of social welfare. \hfill $\blacksquare$
\end{exmpl}


\newcommand{\calP}{\ensuremath{\mathcal{P}}}


\noindent \textbf{Structure of Modest MBB Equilibria. $\;$}
Let us characterize the set of price vectors of modest MBB equilibria, which we denote by $\mathcal{P} = \{ \vecp \mid (\vecx,\vecp) \text{ is modest MBB equilibrium }\}$. We consider the coordinate-wise comparison, i.e., $\vecp \ge \vecp'$ iff $p_j \ge p'_j$ for all $j \in G$. 
\begin{theorem}
	\label{thm:lattice}
  The pair $(\calP, \ge)$ is a lattice.
\end{theorem}
Given $\vecp$ and $\vecp'$, we partition the set of goods into three sets: $S_0 = \{j \mid p_j = p'_j\}$, $S_1 = \{j \mid p_j < p'_j\}$ and $S_2 = \{j \mid p_j > p'_j\}$. Let $\Gamma(T, \vecp)=\{i \mid x_{ij} > 0 \textrm{ for some }j \in T\}$ denote the set of buyers who are allocated a nonzero amount of any good in set $T$ in the equilibrium $(\vecx, \vecp)$. The proof exploits the following properties about the sets $S_0, S_1$ and $S_2$.
\begin{lemma}
	\label{lem:lattice}
  Given any two modest MBB equilibria $(\vecx, \vecp)$ and $(\vecx', \vecp')$, we have
  \begin{enumerate}
  \item[$(i)$] $\Gamma(S_i, \vecp) = \Gamma(S_i, \vecp')$ for $i=0,1,2$, i.e., the set of buyers who buy the goods of $S_i$ with respect to prices $\vecp$ and $\vecp'$are same. 
  \item[$(ii)$] $\Gamma(S_0, \vecp), \Gamma(S_1, \vecp)$ and $\Gamma(S_2, \vecp)$ are mutually disjoint.
  \item[$(iii)$] All buyers in $\Gamma(S_1, \vecp)$ and $\Gamma(S_2, \vecp)$ are capped buyers in both $(\vecx, \vecp)$ and $(\vecx', \vecp')$.
  \end{enumerate}
\end{lemma}
\begin{proof}
  We first focus on $S_1$, the set of goods whose prices strictly increase from $\vecp$ to $\vecp'$. For any $i \in \Gamma(S_1,\vecp')$, there is some $j \in S_1$ such that $u_{ij}/p'_j \ge u_{i\ell}/p'_\ell$ for all $\ell \not\in S_1$. Since $u_{ij}/p_j > u_{ij}/p'_j$ and $u_{i\ell}/p'_\ell \ge u_{ij}/p_\ell$ we conclude
  \begin{itemize}
    \item [(a)] $\Gamma(S_1, \vecp') \subseteq \Gamma(S_1, \vecp)$, and 
    \item [(b)] $x_{ij} = 0$ for $i \in \Gamma(S_1, \vecp')$ and $j \notin S_1$. 
  \end{itemize} 
	Next we analyze the total money spent on goods in set $S_1$ from buyers in $\Gamma(S_1, \vecp')$, with respect to equilibria $E=(\vecx, \vecp)$ and $E' = (\vecx', \vecp')$. Due to fact~(b), buyers in $\Gamma(S_1, \vecp')$ will only buy goods in the set $S_1$ in $E$. We have
\begin{equation}
  \label{eq:lat1}
	  \sum_{i \in \Gamma(S_1, \vecp')}M^a_i = \sum_{i \in \Gamma(S_1, \vecp')}\sum_{j \in S_1}x_{ij}p_j \leq \sum_{j \in S_1}p_j\enspace,
\end{equation}
where $M^a_i$ is the money spent by buyer $i \in \Gamma(S_1, \vecp')$ in $E$. If $i$ is uncapped in $E$, he remains uncapped in $E'$ and his spending remains the same. For the other case, if $i$ is capped in $E$, he spends $\sum_{j\in S_1}{x_{ij}p_j}$, and his spending in $E'$ will be no more than $\sum_{j\in S_1}{x_{ij}p'_j}$. Hence, the total increase of spending of buyers in $\Gamma(S_1, \vecp')$ from $\vecp$ to $\vecp'$ will be no more than 
\begin{equation}
  \label{eq:lat2}
  \sum_{i \in \Gamma(S_1, \vecp')}{M^a_i}' - \sum_{i \in \Gamma(S_1, \vecp')}M^a_i = \sum_{i \in \Gamma(S_1, \vecp')}\sum_{j\in S_1}{(x_{ij}p'_j - x_{ij}p_j)} \leq \sum_{j \in S_1}p'_j - \sum_{j \in S_1}p_j\enspace,
\end{equation}
where ${M^a_i}'$ is the money spent by buyer $i \in \Gamma(S_1, \vecp')$ in $E'$. Further, due to the definition of $\Gamma(S_1, \vecp')$ and the fact that $(\vecx', \vecp')$ is a market equilibrium, we have
\begin{equation}
  \label{eq:lat3}
  \sum_{j \in S_1}p_j' \leq \sum_{i \in \Gamma(S_1, \vecp')}{M^a_i}' \enspace.
\end{equation}
Summing up equations~\eqref{eq:lat1}, \eqref{eq:lat2} and~\eqref{eq:lat3} implies that they can hold at the same time if and only if the three inequalities in them are all equalities. In~\eqref{eq:lat1} this implies there is no buyer outside $\Gamma(S_1, \vecp')$ that has nonzero allocation of any good in $S_1$ in equilibrium $E$. Hence we have $\Gamma(S_1, \vecp) = \Gamma(S_1, \vecp')$. In equation~(\ref{eq:lat2}) this implies all buyers in $\Gamma(S_1, \vecp')$ are capped buyers with both $\vecp$ and $\vecp'$. In~\eqref{eq:lat3} this implies there is no buyer in $\Gamma(S_1, \vecp')$ that has nonzero allocation of any good outside $S_1$ in equilibrium $(\vecx', \vecp')$. Hence $\Gamma(S_1, \vecp')$ does not have any overlap with $\Gamma(S_0, \vecp') \cup \Gamma(S_2, \vecp')$.

Reversing the role of $(\vecx, \vecp)$ and $(\vecx', \vecp')$ and using the same argument, we can prove the same claims for set $S_2$. That is, $\Gamma(S_2, \vecp) = \Gamma(S_2, \vecp')$, $\Gamma(S_2, \vecp')$ has no overlap with $\Gamma(S_0, \vecp') \cup \Gamma(S_1, \vecp')$, and all buyers in $\Gamma(S_2, \vecp')$ are capped buyers in both equilibria. Further, the claims for both $S_1$ and $S_2$ implies $\Gamma(S_0, \vecp) = \Gamma(S_0, \vecp')$ and $\Gamma(S_0, \vecp)$ has no overlap with $\Gamma(S_1, \vecp) \cup \Gamma(S_2, \vecp)$.
Together they prove the lemma.
\end{proof}
\begin{proof}[Proof of Theorem~\ref{thm:lattice}]
  It suffices to show that for any two modest MBB equilibria $(\vecx, \vecp)$ and $(\vecx', \vecp')$ the supremum $\overline{\vecp} = \vecp \lor \vecp'$ and infimum $\underline\vecp = \vecp \land \vecp'$ are both in $\calP$.
  
We prove the claim for $\overline\vecp$. The claim for $\underline\vecp$ can be proved very similarly, and we omit the details. Consider the supremum $\overline{\vecp}$ and a suitable allocation $\overline{\vecx}$ given by
\[
\begin{array}{ccc}
\overline{p}_j = 
  \begin{cases}  
    p'_j & \text{if $j \in S_1$}\\
    p_j  & \text{if $j \in S_2$}\\
    p_j \text{(or $p'_j$)} & \text{if $j \in S_0$}\enspace
  \end{cases} 
 & \hspace{1.5cm}\null
 &
\overline{x}_{ij} = 
  \begin{cases}  
    x'_{ij}   & \text{if $j \in S_1$}\\
    x_{ij}  & \text{if $j \in S_2$}\\
    x_{ij} & \text{if $j \in S_0$}\enspace.
  \end{cases}  
\end{array}
\]
We will show that $(\overline{\vecx}, \overline\vecp)$ is a modest MBB equilibrium. Compare $\vecp$ to $\overline\vecp$, we only increase the prices of goods in $S_1$ from $\vecp$ to $\vecp'$. Hence, the equality edges connecting to $S_2$ and $S_0$ remain the same when prices change from $\vecp$ to $\overline\vecp$. Therefore, with regard to price vector $\overline\vecp$, $x_{ij}$ for goods $j \in S_2 \cup S_0$ will remain a feasible MBB allocation that clears all surpluses of goods in $S_2 \cup S_0$. Using a similar argument, one can show that $x'_{ij}$ for goods $j \in S_1$ will remain MBB and clear all surpluses of goods in $S_1$. We conclude that $(\overline\vecx, \overline\vecp)$ is indeed a modest MBB equilibrium.
\end{proof}
\begin{corollary}
There exists a modest MBB equilibrium with coordinate-wise highest (resp.\ lowest) prices. It yields the maximum (resp.\ minimum) revenue for the seller among all modest MBB equilibria.
\end{corollary}
\setcounter{exmpl}{1}
\begin{exmpl} Consider the following market with two buyers and two goods. Let $u_{11} = u_{12} = u_{22} = 1$ and $u_{21} = 0$. Let $M_1 = M_2 = 1$ and $c_1 = 1$. Then $x_{11} = x_{22} = 1$, $x_{12} = 0$, $p_1 = p_2 = 1$ is a modest MBB equilibrium with maximum revenue. A modest MBB equilibrium with minimum revenue has the same allocation and $p_1 = 0$ and $p_2 = 1$. 
\end{exmpl}

\newcommand{\norm}[1]{\| #1 \|}
\newcommand{\twonorm}[1]{\norm{#1}_2}
\newcommand{\tr}{\tilde{r}}
\newcommand{\tM}{\tilde{M}}
\newcommand{\tp}{\tilde{p}}
\newcommand{\nfrac}{\nicefrac}
\newcommand{\0}{\mbox{\boldmath $0$}}

\section{Computing a Modest MBB Equilibrium with Maximum Revenue}
\label{sec:algorithm}

In this section, we describe an efficient algorithm to compute a modest MBB equilibrium. In fact, we compute the one with coordinate-wise highest prices and maximum revenue among all modest MBB equilibria. Define the {\it active budget} of buyer $i$ at prices $\vecp$ as $M^a_i = \min\{M_i, c_i/\alpha_i\}$, where $\alpha_i = \max_{j \in G} u_{ij}/p_j$ is the MBB ratio. The active budget of buyer $i$ is the minimum of $M_i$ and the minimum amount of money needed to buy a bundle of goods with utility $c_i$. If $M^a_i = c_i/\alpha_i$ then buyer $i$ is capped, otherwise uncapped. 

\subsection{Flow Network and Initialization}
Given prices $\vecp$, let $A=\{(i,j) \subseteq B \times G \mid u_{ij}/p_j = \alpha_i\}$ be the set of equality edges, and the bipartite graph $(B\cup G, A)$ be the {\em equality graph}. We set up the following flow network $N_p$ using the equality graph by adding a source $s$ and sink $t$. It has nodes $\{s,t\} \cup B \cup G$ and edges $(s,i)$ for $i \in B$, $(j,t)$ for $j \in G$ and the equality edges. The edge $(s,i)$ has capacity $M^{a}_i$, and the edge $(j,t)$ has capacity $p_j$. The equality edges have infinite capacity. The flow in the network corresponds to money. We will maintain the following invariants throughout the algorithm.\medskip

\noindent{\bf Invariants:}
\vspace{-0.15cm}
\begin{itemize} 
\item The edges out of $s$ are saturated.
\item Prices and active budgets never increase.
\item Total utility of a buyer never decreases. Once a buyer is capped, it stays capped.
\end{itemize}

We initialize the prices to large values, namely $p_j = \sum_i M_i$. Wlog we will assume that all budgets, caps, and utilities are integers. 

The surplus (residual capacity) of good $j$ is $r_j = p_j - f_{jt}$, where $f_{jt}$ is the flow from good $j$ to $t$. Then $1 - f_{jt}/p_j$ is the
fraction of good $j$ that is not sold. We also keep track of the allocations $x_{ij}$. There might be prices equal to zero and then the allocation cannot be computed from the money flow. Goods that have price zero have no surplus. There is no money flowing through them, although they may be (partially) allocated. 

A subset $T$ of buyers is called {\em tight} with respect to prices $\vecp$ if $\sum_{i \in T} M^a_i = \sum_{j \in \Gamma(T)} p_j$, where $\Gamma(T) \subseteq S$ is the set of goods connected to $T$ in the equality graph.

A {\em balanced flow} is a maximum flow in $N_p$ which minimizes the 2-norm of surplus vector $r$. Let $|r| = |r_1| + \ldots + |r_n|$ and $\norm{r} =
(r_1^2 + \ldots + r_n^2)^{1/2}$ be the $\ell_1$ and $\ell_2$ norm of $r$, respectively. 

\subsection{The Algorithm}

\begin{figure*}[t]
\centerline{\framebox{\parbox{5.0in}{
\begin{tabbing}
555\=555\=555\=555\=\kill
\>{\bf Input:} A market with a set of buyers $B$ and a set of goods $G$; \\
\>\>\> Budget $M_i$, happiness cap $c_i$, and utility parameters $u_{ij},\ \forall i\in B, j\in G$; \\
\>{\bf Output:} Equilibrium prices $\vecp$, allocation $x$; \\ [0.3em]
\> $n \ot |B|;\ \ m \ot |G|; \ \ U \ot \max_{i\in B,j\in G}\{M_i, c_i, u_{ij}\}; \ \ \epsilon \ot {1}/((m+n)U^{4(m+n)})$; \\ 
\>Initialize price $p_j \ot \sum_i M_i$ for each good $j$; \\
\>Initialize active budget $M^a_i \ot \min\{M_i, \min_j c_ip_j/u_{ij}\}$ for each buyer $i$; \\ 
\> $f_{ij} \ot 0,\ x_{ij} \ot 0,\ \forall i \in B, j \in G$; \\[0.3em] 
\>{\bf Repeat} \ // phase \\
\>\> $f \ot$ balanced flow in $N_p$;\ \ $x_{ij} \ot f_{ij}/p_j$ if $p_j \neq 0$; \ $r_j \ot p_j - f_{jt}$; \\
\>\> $\delta \ot \max_j r_j$; Pick a good $j$ with surplus $\delta$; \\
\>\> $S \ot \{j\} \cup \{k \in G \mid k \text{ can reach } j \text{ in the residual network w.r.t. } f \text{ in } N_p \setminus \{s, t\}\}$; \\
\>\>{\bf Repeat} \ // iteration\\
\>\>\>  $B' \ot $ Set of buyers who have incident equality edges to $S$; \\
\>\>\> $B'_c \ot$ Set of capped buyers in $B'$ (a buyer $i$ is capped if $M^a_i = \min_j \nfrac{c_ip_j}{u_{ij}}$); \\
\>\>\> $B'_u \ot B'\setminus B'_c$ (set of uncapped buyers); \\
\>\>\> $x\ot 1$; Define prices and active budgets as follows: \\
\>\>\>\> $p_j \ot xp_j,\  \forall j \in S$;\ $M^a_i \ot xM^a_i,\ \forall i \in B_c'$; \\ 
\>\>\> Decrease $x$ continuously down from 1 until one of the following events occurs\\[0.3em]
\>\>\> {\bf Event 1:} An uncapped buyer becomes capped \\ 
\>\>\> {\bf Event 2:} A new equality edge appears \\ 
\>\>\>\> Recompute $N_p$;\\
\>\>\>\> $f \ot $ balanced flow in $N_p$; \ $x_{ij} \ot f_{ij}/p_j$ if $p_j \neq 0$; \\
\>\>\>\> $S \ot S \cup \{j \in G \mid j \mbox{ can reach } S \mbox{ in the residual network w.r.t. } f \mbox{ in } N_p \setminus \{s, t\}\}$; \\[0.3em]
\>\>\> {\bf Event 3:} A subset of $B'$ becomes tight // phase ends\\[0.3em]
\>\>{\bf Until} Event 3 occurs; \\[0.3em]
\>{\bf Until} $|r| \le \epsilon$;\\
\> Recompute $N_p$;\\
\> $f \ot $ balanced flow in $N_p$; \ $x_{ij} \ot f_{ij}/p_j$ if $p_j \neq 0$;
\end{tabbing}\vspace{-1em}
}}}\caption{The complete algorithm}\label{fig:algo}
\end{figure*}

The complete algorithm is shown in Figure~\ref{fig:algo}. We initialize price $p_j$ of each good $j$ to $\sum_i M_i$. This ensures that the invariants are satisfied, namely a maximum flow in network $N_p$ saturates all edges out of $s$. We initialize every active budget $M^a_i = \min\{M_i, \min_j c_ip_j/u_{ij}\}$, and flow $f$ and allocation $x$ equal to zero.

The algorithm is divided into a set of phases, and each phase is further divided into a set of iterations. A phase starts with the computation of a balanced flow in $N_p$. Let the surplus of good $j$ be $p_j - f_{jt}$. We pick a good $j$ with maximum surplus, and we compute a set of goods $S$ containing $j$ and the goods which can reach $j$ in the residual network corresponding to $N_p$ without using nodes $s$ and $t$. The surplus of each good in $S$ is the same, and maximum among all goods. We denote by $B'$ the set of buyers who have equality edges to goods in $S$, and by $B'_c$ and $B'_u$ the sets of capped and uncapped buyers in $B'$, respectively. Note that $x_{ij} = 0$ for all $i \in B'$ and $j \not\in S$, since $x_{ij} > 0$ would imply $j \in S$. 

We begin with an iteration, where we use a factor $x$ to set the price of each good $j\in S$ to $x p_j$ and the active budget of each buyer $i \in B'_c$ to $xM^a_i$. The prices and active budgets of the remaining goods and buyers remain unchanged. We decrease $x \le 1$ continuously until some structural change happens. Our goal here is to decrease prices as much as possible. By changing prices in this manner, all the equality edges between $B'$ and $S$ stay intact and the equality edges between $B'$ and $G\setminus S$ become non-equality.

A possible structural change is that an uncapped buyer becomes capped. When a buyer $i \in B'$ is uncapped, $M_i < \min_j c_i p_j/u_{ij}$. Prices are decreasing, so this may become an equality. We term the first such change Event 1. Then we move buyer $i$ from $B'_u$ to $B'_c$. 

Another possible change is that a new equality edge appears from a buyer in $B \setminus B'$ to a good in $S$. Prices of goods in $S$ are decreasing, so goods in $S$ are becoming attractive to buyers outside $B'$. Note that there cannot be a new equality edge from a buyer in $B'$ to a good outside $S$. We term the first such change Event 2. Then we recompute the flow network $N_p$ and a balanced flow in $N_p$. Next, we compute the set $S'$ of goods $j \in G \setminus S$ that can reach a good in $S$ in the residual graph corresponding to $N_p$ without using the nodes $s$ and $t$. Due to the property of balanced flows, the surplus of each good in $S'$ is at least the surplus of some good in $S$. Finally, we add goods in $S'$ to $S$. 

Apart from the structural changes, we also maintain the invariants. The only invariant that can become violated with these changes is that the edges out of $s$ are saturated. Hence, we need to stop when a subset $T$ of $B'$ becomes tight. Clearly, if prices are decreased further, then buyers in $T$ will not be saturated, so we stop decreasing prices at this stage. We term this Event 3, and then the phase ends. We show in Lemma~\ref{lem:gd} below that during a phase, the 2-norm of the surplus vector decreases geometrically. 
The last phase ends when the total surplus becomes tiny. In fact, we will show that the surplus is actually zero at this point. We recompute a balanced flow and terminate. 


%

When the prices of a set of goods hit zero in an iteration of the algorithm, then we do not change the allocation of these goods, and all the buyers interested in these goods must be capped. Since each buyer gets a modest allocation before the prices hit zero, the same allocation remains modest. None of the goods in the set is completely allocated. We delete these goods and the buyers to which they are allocated from consideration.

\newcommand{\sset}[1]{\{ #1 \}}
\setcounter{exmpl}{0}

\begin{exmpl}[continued] 
Consider our algorithm applied to the example market above. We initialize $p_1$ and $p_2$ to $M_1 + M_2 = 4$. The active budgets become $M_1^{a} = \min_j
c_1 p_j/u_{ij} = 4/5$ and $M_2^{a} = 1$. The edges $(1,1)$ and $(2,1)$ are equality edges and the balanced flow is $f_{11} = 4/5$, $f_{21} = 1$, and
$f_{12} = f_{22} = 0$. The surpluses are $r_1 = 4 - 9/5 = 11/5$ and $r_2 = 4 - 0 = 4$. Thus $S = \sset{2}$. We decrease $p_2$ to $xp_2$. At $x = 1/2$,
the edge $(2,2)$ becomes an equality edge. Now $p_2 = 2$. The balanced flow does not change and hence $r_1 = 11/5$ and $r_2 = 2$. Thus $S = \sset{1}$.
We decrease $M_1^{a}$ to $4x/5$ and $p_1$ to $4x$. At $x = 5/16$, $B'$ becomes tight. We now have $M_1^{a} = 1/4$ and $p_1 = 5/4$. The balanced flow is
$f_{11} = 1/4$ and $f_{21} = 1$. Thus $r_1 = 0$ and $r_2 = 2$. So $S = \sset{2}$. We change $p_2$ to $p_2 x$. At $x = 5/16$, the edge $(2,2)$ becomes
an equality edge. Now $p_2 = 5/8$. The balanced flow is $f_{11} = 1/4$, $f_{21} = 11/16$, and $f_{22} = 5/16$. Then $r_1 = r_2 = 5/16$. Thus $S =
\sset{1,2}$. We now decrease $M_1^{a}$ to $x \cdot 1/4$, $p_1$ to $5x/4$ and $p_2$ to $5x/8$. At $x = 8/13$, $B'$ becomes tight and we have $p_1 = 10/13$, $p_2 = 5/13$, $M_1^{a} = 2/13$, $x_{11} = 1/5$, $x_{21} = 4/5$, $x_{22} = 1$, $f_{11} = 2/13$, $f_{21} = 8/13$, and $f_{22} = 5/13$. \hfill $\blacksquare$
\end{exmpl}

\subsection{Analysis}
\begin{lemma}\label{lem:inv}
The invariants hold during the run of the algorithm. 
\end{lemma}
\begin{proof}
Clearly, prices are non-increasing. As a result, the active budgets of buyers are also non-increasing. The tight-set event makes sure that buyers are always saturated. As a result, the total utility of each buyer never decreases, since he spends his entire active budget and prices are non-increasing. 
\end{proof} 
Phases consist of iterations, which end with Event 1, 2, or 3. A phase ends with Event 3. 
\begin{lemma}\label{lem:numiter}
Each phase has at most $2n$ iterations.
\end{lemma}
\begin{proof}
Each iteration ends with one of the three events. In case of Event 1, an uncapped buyer becomes capped, and there can be at most $n$ iterations of this kind (Lemma~\ref{lem:inv}). In case of Event 2, a new equality edge arises from a buyer outside $B'$ to a good in $S$. This adds at least one new buyer to $B'$. Further prices are changed in a way so that no buyer leaves $B'$, hence the number of such events are again at most $n$. In case of Event 3, the phase ends. 
\end{proof} 

Our next goal is to show that the 2-norm of the surplus vector decreases substantially during a phase. Let $r$ and $r'$ be the surplus vectors at the beginning and at the end of a phase respectively. For the purpose of our analysis we also maintain an intermediate flow $f$ continuously as we change prices in each iteration; this flow is \emph{not} maintained by the algorithm.  When we recompute a balanced flow during Event 2, then $f$ will be reset to the balanced flow. It is defined as 
$\forall i\in B'_c:  f_{ij} \ot xf_{ij} \text{ and } \forall i\in B'_u:  f_{ij} \ot f_{ij}.$
$f$ ensures that all buyers are saturated. If the surplus of a good $j$ becomes zero corresponding to $f$, then we keep its surplus equal to zero and reroute extra flow from $j$ to some other good with positive surplus, using a path in the residual network corresponding to $f$. If there is no such path, then this implies Event 3 has occurred, in which case the current phase is done. Consider an intermediate iteration $t$. With respect to $f$, let $r^t = (r^t_1, \dots, r^t_m)$ be the surplus vector at the beginning of iteration $t$, and let $\tr^t = (\tr^t_1, \dots, \tr^t_m)$ be the surplus vector before we recompute a balanced flow in iteration $t$ if Event 2 occurs. 
\begin{lemma}\label{lem:rdec}
$\tr_j^t \le r^t_j,\ \forall j \in G$, and $\norm{r^{t+1}} \le \norm{\tr^t} \le \norm{r^t}$. 
\end{lemma} 
\begin{proof}
For the first part, let $p^t$ and $f^t$ respectively denote prices and flows at the beginning of iteration $t$. Note that both prices and flows are unchanged for goods outside $S$, hence $\tr_j^t = r^t_j, \forall j \not\in S$. For goods in $S$ we have $r_j^t = p_j^t - \sum_{i\in B_c'} f_{ij}^t - \sum_{i\in B_u'} f_{ij}^t$ and $\tr^t_j = xp_j^t - x\sum_{i \in B_c'} f_{ij}^t - \sum_{i\in B_u'} f_{ij}^t$ for an $x\le 1$, which implies that $\tr_j^t \le r^t_j,\ \forall j \in S$ before surplus of some good becomes zero. Further, when the surplus of a good $j$ becomes zero, we reroute extra flow from $j$ to some other good $k$. This will further decrease the surplus of $k$. This proves the first part.

Due to the first part, we have $\norm{\tr^t} \le \norm{r^t}$. For the first inequality of the second part, note that $r^{t+1}$ is different from $\tr^t$ only due to recomputation of the balanced flow. The flow at $\tr^t$ is feasible, and the balanced flow can only make the norm better than the norm at $\tr^t$. The second part follows.
\end{proof}
\begin{lemma}\cite{DevanurPSV08}\label{lem:decsur}
Suppose $f$ and $f^*$ are a feasible and a balanced flow in $N_p$, resp., and $r$ and $r^*$ are the surplus vectors w.r.t.\ $f$ and $f^*$, resp. If $r^*_j = r_j - \delta$ for some good $j$ and $\delta >0$, then $\norm{r^*}^2 \le \norm{r}^2 - \delta^2$.
\end{lemma} 
\begin{lemma}\label{lem:gd}
$\norm{r'}^2 \le (1-\nfrac{1}{4mn})\norm{r}^2$. 
\end{lemma}
\begin{proof}
Consider the value of $\gamma = \min\{r_j\ |\ j \in S\}$ during a phase. When the phase begins, $\gamma =\delta$, and when it ends $\gamma=0$. Recall that $S$ only grows, and when we add a new good $k$ to $S$, then the surplus of $k$ is at least the surplus of some good already in $S$. This implies that $\gamma$ does not change when we add new goods to $S$. 

Let $t_1, \dots, t_l$ be the iterations where $\gamma$ decreases, and let $\delta_i > 0$ be the amount of decrease in iteration $t_i$. Further we break each $\delta_i$ into two parts $\delta_{i1}$ and $\delta_{i2}$ such that $\delta_i = \delta_{i1}+\delta_{i2}$. Here $\delta_{i1}$ is the amount of decrease due to the flow change before we recompute balanced flow, and $\delta_{i2}$ is the amount of decrease due to recomputation of balanced flow. Next consider only positive $\delta_{i1}$'s and $\delta_{i2}$'s. Clearly, $l \le 2n$ and $\sum_{i: \delta_{i1} > 0, \delta_{i2} >0} (\delta_{i1} + \delta_{i2}) \ge \delta$. Using Lemmas~\ref{lem:rdec} and~\ref{lem:decsur}, we have 
$\norm{r'}^2 \le \norm{r}^2 - (\delta_{11}^2 + \delta_{12}^2 + \dots + \delta_{l1}^2 + \delta_{l2}^2) \le \norm{r}^2 - \delta^2/4n.$
Since $\norm{r}^2 \le m\delta^2$, we have $\norm{r'}^2 \le (1-\nfrac{1}{4mn})\norm{r}^2$. 
\end{proof}

{\noindent \textbf{Polynomial Running Time.} $\;$}
In each iteration, the prices of goods in $S$ are multiplied by a value that itself depends on the prices. It is not obvious why the size of the numbers in the computation is polynomially bounded. Here we show that the sizes of intermediate prices and flows in our algorithm remain polynomially bounded. 
\begin{lemma} 
	\label{lem:connectFlow}
	All goods in $S$ are connected by equality edges at all times. There is no flow from buyers in $B'$ to goods outside $S$. 
\end{lemma}
\begin{proof}
In each phase, $S$ is initialized to all goods that can reach the selected good of highest surplus in the residual graph. The set $B'$ is always the set of buyers that have equality edges to $S$. The prices of the goods in $S$ are changed by the same factor and hence no equality edges in $B' \times S$ is destroyed. When a good is added to $S$, it has a path to $S$ in the residual graph and hence is connected to $S$ via equality edges.

Buyers in $B'$ are connected to goods in $S$ by equality edges. Assume there would be a flow from a buyer in $B'$ to a good $g$ outside $S$. Then there would be a path from $g$ to $S$ in the residual graph and hence $g$ would belong to $S$. 
\end{proof}
Cap-events occur only at a cap-event prices. A cap-event price is any price $p$ with $M_i = c_i p/u_{ij}$ for some $i$ and $j$. Let $P_c = \set{M_i u_{ij}/c_i}{1 \le i \le n, 1 \le j \le m}$. 

Let $A'$ be any subset of the edge set with positive utility such that the graph formed by it is connected. Let $B'$ and $G'$ be the buyers and goods in this connected graph. The prices in the component have only one degree of freedom, i.e., we can select one of the prices, say $p$, as a base price and express any other price in the component as $\alpha p$, where $\alpha$ is a rational whose numerator and denominator are products of at most $m$ utilities. Consider an arbitrary partition of $B'$ into capped buyers $B'_c$ and uncapped buyers $B'_u$; $B'_u$ must be nonempty. The budget of a capped buyer is of the form $c \alpha p$, where $c$ is a cap and $\alpha$ is as above. If there are no surpluses, $p$ must satisfy that (budget of capped buyers + budget of uncapped buyers) equals sum of the prices (in the component). We call a price that can be obtained in this way a submarket price; note that not all submarket prices can actually occur. Let $P_m$ be the set of \emph{submarket prices}.

Let $P_i$ be the set consisting of the initial price and zero. 
A price is \emph{1-linked} if it is of the form $(U_1/U_2)p$ where $p \in P_c \cup P_m \cup P_i$ and $U_1$ and $U_2$ are products of at most $n$
utilities each. A price is \emph{2-linked} if it of the form $(U_1/U_2) p$, where $p$ is 1-linked and $U_1$ and $U_2$ are products of at most $n$ utilities each.
\begin{lemma}\label{lem:bound}
Assuming that all budgets, happiness caps and utilities are integers bounded by $U$, 1-linked and 2-linked prices are rational numbers
whose bit-length is at most $\log{(m+n)}+3(m+n)\log{U}$. 
\end{lemma}
\begin{proof} 
The prices in sets $P_i$ and $P_c$ are clearly rational numbers with bit length at most $n\log nU$. For the set $P_m$, we have the following linear equation in $p$
\[ \sum_j \alpha_j p - \sum_i c_i \alpha_i p = \sum_i M_i,\] 
where both $\alpha_j$'s and $\alpha_i$'s are rational numbers whose numerator and denominator are product of at most $m$ utilities. By simplifying the above equation, we obtain that $p$ is a rational number whose both numerator and denominator are at most $(m+n) U^{m+n}$. This implies that the 1-linked and 2-linked prices are rational numbers whose bit lengths are at most $\log{(m+n)}+3(m+n)\log{U}$. 
\end{proof}
\begin{lemma}
\label{1-2-link-lemma}
At the beginning of a phase, all prices are 1-linked. During a phase, prices outside $S$ are 1-linked. At the end of each iteration, prices in $S$ are 2-linked. 
\end{lemma}
\begin{proof}
At the beginning of the first phase, all prices are equal to the initial price and hence 1-linked. Consider any later phase. Prices of goods outside $S$ are the same as at the beginning of the phase and hence are 1-linked by induction hypothesis. Consider next the prices of the goods in $S$. The
phase ended because some set $T \subseteq B'$ became tight. If $T$ contains an uncapped buyer, the new prices of the goods in $\Gamma(T)$ are submarket prices, where $\Gamma(T)$ is the set of goods connected to $T$ in the equality graph. Since all goods in $S$ are connected to a good in $\Gamma(T)$ by equality edges, all prices in $S$ are 1-linked. If $T$ contains only uncapped buyers and the tight set event co-occurred with a capping event, one of the prices in $S$ will be a cap-event price and hence all other prices in $S$ are 1-linked. The final possibility is that all prices in $S$ are zero. Then they are also 1-linked. This proves that in the beginning all prices are 1-linked.

Now, during a phase prices outside $S$ do not change and hence are 1-linked. A tight-set event ends a phase and after the recomputation of the prices, all prices are 1-linked as shown above. When a capping event occurs, some price of a good in $S$ becomes a cap-event price and hence all prices in $S$ are 1-linked. After a new edge event, all goods in $S$ are connected to some good previously outside $S$ by equality edges. Since the price of the good outside $S$ is 1-linked, the prices in $S$ will be 2-linked.
\end{proof}
\begin{remark}
Why are we not simply stating that all prices are connected to a price in $P_c \cup P_m \cup P_i$ by a sequence of equality edges? During the course of the algorithm, we are loosing equality edges, namely edges connecting buyers in $B'$ to goods outside $S$. Since these edges do no carry flow, this does no harm. However, some of these edges may have played a role in expressing a price in terms of a reference price. Thus we cannot say that current prices are linked to reference prices through current equality edges. We can only say that current prices are linked to reference prices through a path of edges of nonzero utility. We do not know whether these paths stay simple. Lemma~\ref{1-2-link-lemma} shows that the paths are at most $2n$ in length.
\end{remark}

\begin{theorem}
\label{thm:me}
The algorithm in Figure \ref{fig:algo} computes a modest MBB equilibrium. 
\end{theorem}

\begin{proof} 
When the algorithm terminates, we claim that at this stage total surplus $\sum_j r_j = 0$. This will imply that the algorithm in Figure~\ref{fig:algo}
computes a market equilibrium. Consider any good $j$ and the component $C$ of the equality graph containing good $j$. The total surplus $\sum_{j \in C} r_j$ in the component is  $\sum_{j \in C} p_j - \sum_{i\in C} M^a_i$. This is non-negative and less than $\epsilon$. All prices and active budgets of capped buyers can be expressed in terms of one price variable $p$ using equality relations. By Lemma~\ref{lem:bound}, $p$ is a rational number with bounded denominator, and the above inequalities imply $\sum_{j \in C} r_j = 0$. Thus $r_j = 0$ for all $j$. 
%
\end{proof} 
Let $x_{ts}$ denote the value of $x$ when Event 3 occurs in the algorithm. Next we show that $x_{ts}$ can be computed using at most $n$ max-flow computations. This is a generalization of a procedure in~\cite{DevanurPSV08} for computing tight set in case of {\em linear} Fisher markets. 

\begin{lemma}
\label{lem:ts}
$x_{ts}$ can be computed using at most $n$ max-flow computation. 
\end{lemma}
\begin{proof}
Let $C$ be the connected component of the equality graph containing the goods $S$, and consider buyers $B'$ as defined in Figure~\ref{fig:algo}. Let $P$ be the total price of the goods in $S$ (when $x = 1)$. The active budget of the buyers in $B'$ is $U + xV$. $B'$ is not tight at $x = 1$. It goes tight at $x$ determined by $U + xV = Px$, i.e., $x = U/(P - V)$. Let us set the prices of the goods in $S$ to $x \vecp$, where $\vecp$ is the price vector. Also set the active budgets as determined by $x$. Compute a max-flow. 

If all of the budget can be routed, then $B'$ is the smallest tight set. Otherwise, let $(s \cup B_1 \cup G_1, B_2 \cup G_2 \cup t)$ be the minimum cut. At the price vector $xp$, the buyers in $B_2$ can still get rid of their entire budget. Hence the tightest set is contained in $B_1$. We recurse on $s \cup B_1 \cup G_1 \cup t$. This clearly requires at most $n$ max flow computations. 
\end{proof}

\begin{remark}
Whenever $C$ is enlarged, the tightest set must be recomputed. The critical $x$ may decrease. It cannot increase as the buyers in the old $C$ are not acquiring new equality edges. 
\end{remark}

{\noindent \textbf{Maximum Revenue.} $\;$}
Finally, we show that our algorithm gives a modest MBB equilibrium with maximum revenue among all modest MBB equilibria.
\begin{lemma}
	\label{lem:maxPrice}
  Consider the price vector $\vecp$ at the end of any phase of the algorithm. We have $\vecp \ge \vecp'$ for any price vector $\vecp' \in \calP$ of a modest MBB equilibrium.
\end{lemma}
\begin{proof}
Let $T$ be the set of goods whose price strictly increases from $\vecp$ to $\vecp'$. We will prove $T = \emptyset$ in two steps:\\

\noindent{\bf Step 1.} First we show that for any good $j$ with $r_j>0$, we have $j \notin T$. Suppose by contradiction that this is not true. Let $f, x$ and $r$ be the balanced flow, allocation and surplus vector with respect to prices $\vecp$, respectively. Let $\Gamma(T, \vecp)$ and $\Gamma(T, \vecp')$ be the set of buyers connected to $T$ in the equality graph with regard to $\vecp$ and $\vecp'$, respectively. 

    We observe two facts: 
    \begin{itemize}
    \item [(1)] $\Gamma(T, \vecp') \subseteq \Gamma(T, \vecp)$.
    \item [(2)] $x_{ij} = 0$ for $i \in \Gamma(T, \vecp')$ and $j \notin T$. 
    \end{itemize}

 Now we analyze the total money spent on goods in set $T$ from buyers in $\Gamma(T, \vecp')$, with respect to prices $\vecp$ and $\vecp'$ respectively. 
    
Because $r$ is a non-negative vector with at least one positive entry in $T$, we have 
    \begin{equation}
      \label{eq:p1}
      \sum_{i \in \Gamma(T, \vecp')}M^a_i = \sum_{i \in \Gamma(T, \vecp')}\sum_{j \in T}x_{ij}p_j < \sum_{j \in T}p_j.
    \end{equation}

    For every $i \in \Gamma(T, \vecp')$, if $i$ is uncapped with $\vecp$, he remains uncapped with $\vecp'$, and his active budget will remain unchanged. If $i$ is capped with $\vecp$, his active budget with prices $\vecp$ is $\sum_{j\in T}{x_{ij}p_j}$, and his new active budget with prices $\vecp'$ will be no more than $\sum_{j\in T}{x_{ij}p'_j}$. Hence the total increase of all active budgets for buyers in $\Gamma(T, \vecp')$ from $\vecp$ to $\vecp'$ will be no more than 
    \begin{equation}
      \label{eq:p2}
      \sum_{i \in \Gamma(T, \vecp')}\sum_{j\in T}{(x_{ij}p'_j - x_{ij}p_j)} \leq \sum_{j \in T}p'_j - \sum_{j \in T}p_j
    \end{equation}

    Combine inequality~(\ref{eq:p2}) with~(\ref{eq:p1}) indicates that the total active budgets of buyers in $\Gamma(T, \vecp')$ is strictly less than $\sum_{j \in T}p'_j$, and it is not enough to clear all the surpluses of goods in $T$. This contradicts the fact that $\vecp'$ is an equilibrium.\\
  
\noindent{\bf Step 2.} Next we show for any good $j$ with $r_j = 0$, $j$ cannot be in $T$ either. Let $k$ be the last phase at the end of which has $r_j > 0$. Notice that in phase $k+1$, we decrease the price of good $j$ continuously, and the first moment when $r_j$ reaches zero marks the end of this phase. Thus at the end of phase $k+1$, we will also have $p_j \geq p'_j$. In the remaining steps of the algorithm the price of good $j$ is never touched again. This means $p_j \geq p'_j$ will always hold.
\end{proof}
For the main result in this section, assume that all budgets, happiness caps and utilities are integers bounded by $U$. 

\begin{theorem}
\label{thm:rt}
The algorithm in Figure~\ref{fig:algo} computes a modest MBB equilibrium with maximum revenue in $O(mn^6(\log(m+n) + (m+n)\log U))$ time. 
\end{theorem}

\begin{proof} 
In the beginning, the 2-norm of surplus vector $r$ satisfies $\norm{r}^2 \le mn^2U^2$. By Theorem~\ref{thm:me}, the algorithm will terminate before the norm becomes $\norm{r'}^2 = 1/m(m+n)^2U^{8(m+n)}$. Let $k$ denote the number of phases when the surplus becomes $r'$. From Lemma~\ref{lem:gd}, we have $\norm{r}^2(1-1/4mn)^k = \norm{r'}^2$, which implies that the total number of phases in the algorithm is $O(mn(\log(m+n) + (m+n)\log U))$.  

In each phase, we have at most $2n$ iterations, and in each iteration we need to compute the maximum $0\le x\le 1$ when one of the three events occurs. Let $x_c, x_{eq}$ and $x_{ts}$ respectively denote the maximum value of $x$ where Event 1, 2 and 3 occurs. Clearly, $x_c$ can be obtained in $O(n)$ time, $x_{eq}$ can be obtained in $O(mn)$ time, and $x_{ts}$ can be obtained using at most $n$ max-flow computations due to Lemma~\ref{lem:ts}. Further, we recompute a balanced flow in case of Event 2 which further requires at most $n$ max-flow computations~\cite{DevanurPSV08}. Since a max-flow can be obtained in $O(n^3)$ time, each iteration can be implemented in $O(n^4)$ time.  Hence, the total running time of the algorithm is $O(mn^6(\log(m+n) + (m+n)\log U))$. 
\end{proof} 

We conjecture that the running time in Theorem~\ref{thm:rt} can be reduced by a factor of $\tilde{O}(n^2)$ using the perturbation technique from~\cite{DuanGM16} which requires a max-flow to be computed only in a network with forest structure. We have not worked out the details.

%

\section{Computing a Modest MBB Equilibrium with Minimum Revenue}
\label{sec:minPrice}

In this section, we show how to transform in polynomial time any modest MBB equilibrium into one with minimum revenue using the postprocessing procedure in Fig.~\ref{fig:postprocessing}.
\begin{figure*}[t]
\centerline{\framebox{\parbox{5.0in}{
\begin{tabbing}
555\=555\=555\=555\=\kill
\>{\bf Input:} A market with a set of buyers $B$ and a set of goods $G$; \\
\>\>\> Budget $M_i$, happiness cap $c_i$, and utility parameters $u_{ij},\ \forall i\in B, j\in G$; \\
\>\>\> Any modest MBB equilibrium $(\vecx, \vecp)$;\\
\>{\bf Output:} A modest MBB equilibrium $(\vecx,\vecp)$ with minimum revenue;\\ [0.3em]
\>Initialize active budget $M_i' \ot \min\{M_i, \min_j c_ip_j/u_{ij}\}$ for each buyer $i$; \\[0.3em] 
\>$S \ot \{j | \textrm{$p_j>0$ and $j$ does not have incident equality edges to any uncapped buyer}\}$; \\[0.3em]
\>{\bf While}\ $S \neq \emptyset$\\[0.3em]
\>\> $B' \ot $ Set of buyers who have incident equality edges to $S$; \\
\>\> $x\ot 1$; Define prices and active budgets as follows: \\
\>\>\> $p_j \ot xp_j,\  \forall j \in S$;\ $M^a_i \ot xM^a_i,\ \forall i \in B_c'$; \\ 
\>\> Decrease $x$ continuously down from 1 until one of the following events occurs\\[0.3em]
\>\> {\bf Event 1:} $x$ becomes zero; \\ 
\>\> {\bf Event 2:} A new equality edge appears \\ 
\>\> Recompute $N_p$ and $S$;\\ 
\> {\bf EndWhile}
\end{tabbing}\vspace{-1em}
}}}\caption{The postprocessing algorithm for an equilibrium with minimum revenue}\label{fig:postprocessing}
\end{figure*}
\begin{theorem} 
The algorithm in Figure~\ref{fig:postprocessing} computes a modest MBB equilibrium with minimum revenue.
\end{theorem}

\begin{proof}
  It is easy to check that throughout the algorithm, $(\vecx, \vecp)$ always remains a modest MBB equilibrium. Assume by contradiction that at the end of the algorithm, $(\vecx, \vecp)$ is not an equilibrium with smallest prices. Let $(\vecx', \vecp')$ be an equilibrium with smallest prices, and define $S_1 = \{j \mid p_j > p'_j\}$. By Lemma~\ref{lem:lattice} property (3), all buyers in $\Gamma(S_1, \vecp)$ are capped buyers. Because prices of goods in set $S_1$ decreases from $\vecp$ to $\vecp'$, every buyer $i$ incident to $S_1$ in the equality graph with prices $\vecp$ will only have equality edges to $S_1$ with prices $\vecp'$. Therefore we have $i \in \Gamma(S_1, \vecp') = \Gamma(S_1, \vecp)$ (the equality is again by Lemma~\ref{lem:lattice}). This implies $\Gamma(S_1, \vecp)$ is also the set of buyers who have incident equality edges to $S_1$ with prices $\vecp$. Hence, set $S$ is nonempty for the While loop and the algorithm should not terminate.
\end{proof}




\section{Extensions}
\label{sec:extend}

In the previous section, we proposed an algorithm for computing a modest MBB equilibrium, which has a Pareto-optimal allocation. When we depart from the set of such equilibria, then utilities in market equilibrium are not uniquely determined. In fact, we show that market equilibria with maximum social welfare might not be modest MBB equilibria, and computing such optimal equilibria becomes \classNP-hard.

\begin{theorem}
	\label{thm:NPC}
	It is \classNP-hard to compute a market equilibrium that maximizes social welfare.
\end{theorem}

\begin{proof}
	We reduce from {\sc 3-Dimensional Matching}. Consider an instance $I$ composed of three disjoint sets $A$, $B$, $C$ of elements and a set $T \subseteq A \times B \times C$ of triples. Let $n = |A| = |B| = |C|$ be the number of elements in each set and $m = |T|$ the number of triples. W.l.o.g.\ assume $m \ge n$. Now we construct a Fisher market based on $I$ as follows. For each element $i \in A \cup B \cup C$ we introduce an \emph{element agent} $i$ with budget 1. For each triple $j \in T$ we introduce a good $j$ and an \emph{auxiliary agent} $i_j$ with budget 1. All these agents have linear utility functions. In addition, there is a single \emph{decision agent} $i_d$ with a budget-additive utility function and a budget of $4m^2(m-n)$.

For the utility values, for each agent $i \in A \cup B \cup C$ we assume $u_{ij} = 1$ if triple $j$ contains $i$ and 0 otherwise. For auxiliary agent $i_j$ the utility is $u_{i_jj} = 1/m^3$ and 0 for all other goods. Finally, the decision agent $i_d$ has utility $u_{i_dj} = 1/m^3$ for every good $j$ and a cap of $c_{i_d} = (m-n)/(m(m^2+1))$. Our claim is that a market equilibrium with social welfare of $W = 3n\cdot(1/4) + n \cdot(1/4m^3) + (m-n)/m^3$ exists if and only if the instance $I$ has a solution.

First, suppose $I$ has a solution $S \subseteq T$. Then we set the prices to be $p_j = 4$ for every $j \in S$ and $p_j = m^2+1$ for every $j \not\in S$. As for the allocation, each agent $i$ spends its entire budget of 1 on the good $j \in S$ that includes him. Each auxiliary agent spends its budget on the corresponding good. Finally, the decision agent $i_d$ spends a budget of $m^2$ on each of the $m-n$ goods $j \not\in S$. Observe that all goods are allocated, and (since w.l.o.g.\ we can assume $m > 2$) every agent with linear utility function spends its entire budget on an MBB good. The decision agent has optimal utility $(m-n)\cdot (1/m^3) \cdot m^2/(m^2+1) = c_{i_d}$. As such, we obtain a market equilibrium. Straightforward inspection reveals that the social welfare in this state is indeed $W$.

On the other hand, assume that a market equilibrium achieves a social welfare of at least $W$. Note that for each good $j$, the auxiliary agent can at most obtain a utility of $1/m^3$ by getting all of good $j$. Similarly, the decision agent can obtain at most $m$ goods and get a utility of $c_{i_d}$ for all of them. Thus, by giving all goods to auxiliary and decision agents, together they can contribute at most $1/m^2$ to the social welfare.

We first observe that in every market equilibrium the decision agent obtains a utility of $c_{i_d}$. Consider any good $j$ and let us broadly overestimate the price in equilibrium by assuming that the auxiliary agent and all three element agents spend a total budget of 4 on $j$. This is clearly an upper bound on the money that is spent by the element and auxiliary agents on good $j$. To derive an upper bound, assume this happens on every good $j$. Even in this case, the decision agent can contribute a budget of $4m^2$ to any set of $m-n$ goods. Since in a market equilibrium, the goods must be shared in proportion to the money spent, the decision agent would thereby be able to obtain a share of $4m^2/(4m^2+4) = m^2/(m^2+1)$ from each good it contributes to. In total this yields a utility of $(m-n) \cdot (m^2/(m^2+1)) \cdot (1/m^3) = c_{i_d}$. Hence, in every market equilibrium the decision agent obtains at least a total share of $(m-n)\cdot m^2/(m^2+1)$ of all the goods. Thus, the total remaining supply of goods that can be allocated to the remaining agents is at most $n + (m-n)/(m^2+1)$. 

Let us now discuss how to distribute this remaining supply optimally among the agents. For every good $j$, any equilibrium allocation must be proportional to the incoming money. We remove the fraction obtained by the decision agent, denote the remaining supply by $s_j$, and note $s_j \ge 0$ and $\sum_j s_j \le n + (m-n)/(m^2+1)$. The auxiliary agent always spends its budget of 1 on $j$. Let $y_j$ be the money spent by element agents on good $j$, so $3 \ge y_j \ge 0$ and $\sum_j y_j = 3n$. The welfare obtained from good $j$ by auxiliary and element agents in any equilibrium is 
\[
	s_j\left(\frac{y_j}{y_j+1} + \frac{1}{y_j+1} \cdot \frac{1}{4m^3}\right)\enspace.
\]
Hence, the social welfare obtained by element and auxiliary agents in any market equilibrium is upper bounded by the optimum solution to the following optimization problem:
\begin{equation}\label{pro:NPC}
	\begin{array}[4]{rrll}
	\text{Max.} & \multicolumn{3}{l}{\D \sum_{j \in [m]} s_j \frac{y_j + 1/(4m^3)}{y_j+1} \vspace{0.4cm}}\\
	\text{s.t.} & \D \sum_{j \in [m]} s_j & = n + \frac{m-n}{m^2+1} & \vspace{0.2cm}\\
							& \D \sum_{j \in [m]} y_j & = 3n & \vspace{0.2cm}\\
 							& y_j \le 3 & \forall j \in [m]\vspace{0.2cm}\\
 							& s_j \le 1 & \forall j \in [m]
	\end{array}
\end{equation}
The objective function is linear in the $s_j$ and convace in the $y_j$, the constraints are concave, the equality constraints are affine and their gradients are linearly independent. The feasible solution $y_j = 3/m$ and $s_j = (n + (m- n)/(m^2 + 1))/m$ satisfies the inequality constraints with strict inequality. Hence, the KKT-conditions characterize the unique optimal solution. We use dual variables $\alpha$ and $\beta$ for the equality constraints, $\lambda_j$ and $\mu_j$ for the inequality constraints. The optimal solution must satisfy
\begin{align*}
- \frac{y_j + 1/(4m^3)}{y_j + 1} + \alpha + \mu_j &= 0\\
- s_j \frac{1 - 1/(4m^3)}{(y_j + 1)^2} + \beta + \lambda_j &= 0\\
\lambda_j (y_j - 3) &= 0 \text{ and } \lambda_j \ge 0 \text{ for all $j$}\\
\mu_j (s_j - 1) &= 0 \text{ and } \mu_j \ge 0 \text{ for all $j$}
\end{align*}
Thus, $\alpha \le  (y_j + 1/(4m^3))/({y_j + 1})$ for all $j$. Note that $\alpha < (y_j + 1/(4m^3))(y_j + 1)$ implies $\mu_j > 0$ and $s_j = 1$. Also $\beta_j \le s_j (1 - 1/(4m^3))/(y_j + 1)^2$ for all $j$. Similarly, $\beta_j <  s_j (1 - 1/(4m^3))/(y_j + 1)^2$ implies $\lambda_j > 0$ and $y_j = 3$. We number the $y$'s such that $y_1 \ge y_2 \ge \ldots \ge y_m$. Let $\ell$ be such that $y_1 \ge \ldots \ge y_{\ell} > y_{\ell+1} = \ldots = y_m$. Then $s_j = 1$ for $1 \le j \le \ell$ and hence $\ell \le n$. Let $k \le \ell$ be such that $y_1 = \ldots = y_k > y_{k+1}$. 
Then $y_j = 3$ for $1 \le j \le k$. If $k < \ell$, $y_{k+2} > 0$ since $\sum_j y_j = 3n$, and we increase the objective by increasing $y_{k+1}$. Thus $k = \ell$. If $\ell < n$, $s_{\ell + 2} > 0$ and we increase the objective by increasing $s_{\ell+1}$ and $y_{\ell +1}$. Thus $\ell = n$, and the unique optimum is $y_1 = \ldots = y_n = 3$, $y_{n+1} = \ldots = y_m = 0$, $s_1 = \ldots = s_n = 1$. 

This proves that in the optimum there are $n$ goods to which the decision player does not contribute ($s_j = 1$) and for which there are exactly three element players that can contribute all their budget to this good ($y_j = 3$). Thus, the upper bound on the social welfare is attained only when the decision player contributes to exactly $m-n$ goods such that the remaining $n$ goods correspond to a partition of the $3n$ agents into $n$ disjoint triples. By straightforward inspection, we see that the upper bound on the social welfare amounts to exactly $W$. A market equilibrium of social welfare $W$ can exist only if there is a solution to the underlying instance $I$. This concludes the proof.
\end{proof}
As a corollary, we note that the proof can also be used to show \classNP-hardness for optimizing any constant norm of utility values.
\begin{corollary}
	\label{cor:NPC}
	It is \classNP-hard to compute a market equilibrium $(\vecx,\vecp)$ that maximizes $\sum_i (u_i(\vecx))^\rho$, for every constant $\rho > 0$.
\end{corollary}
\begin{proof}
For $\rho > 1$, we can use exactly the same reduction. The optimum coincides with the optimum for social welfare, since we still want to maximize the share of goods assigned to the element agents. For constant $0 < \rho < 1$ and sufficiently large $m$, the common factor $1/(4m^3)$ is strong enough to keep the incentive of maximizing the share of the element agents. 
\end{proof}
There are several ways of introducing satiation points into the utility function. Instead of a global cap, let us assume there is a cap $c_{ij}$ for the utility buyer $i$ can obtain from good $j$. A good-based budget-additive utility of buyer $i$ is then $u_i(\vecx_i) = \sum_{j} \min(c_{ij}, u_{ij}x_{ij})$. This variant turns out to be an elementary special case of separable piecewise-linear concave (SPLC) utilities, in which every piece consists of a linear segment followed by a constant segment. We show that even finding a single market equilibrium here becomes \classPPAD-hard.
\begin{theorem}
	\label{thm:PPAD}
	It is \classPPAD-hard to compute a market equilibrium in Fisher markets with good-based budget-additive utilities.
\end{theorem}
\begin{proof}
We adapt the construction of Chen and Teng~\cite{ChenT09} to prove the theorem. They show \classPPAD-completeness of computing an approximate equilibrium in Fisher markets under SPLC utilities where each PLC function has at most two segments. Here, the second segment can have positive rate of utility, i.e., non-zero slope, hence \classPPAD-hardness for Fisher markets under good-based budget-additive utilities where the second segment has zero slope, i.e., no utility, requires adjustment in their construction. 

Chen and Teng~\cite{ChenT09} reduce the \classPPAD-hard problem of computing an approximate Nash equilibrium in a two-player game to the problem of computing an approximate equilibrium in Fisher markets under SPLC utilities. Their main idea is to construct a family of {\em price-regulating} markets $\M_n$ for each $n\ge 1$, which has $n$ buyers and $2n$ goods. In $\M_n$, each buyer has budget of $3$ units and each good has supply of $1$ unit, and every approximate equilibrium price vector $\p$ satisfies the following {\em price-regulation} property: 
\begin{equation}
\label{eqn:prp}
\frac{1}{2} \le \frac{p_{2k-1}}{p_{2k}} \le 2 \ \ \ \ \ \mbox{and} \ \ \ \ \  p_{2k-1} + p_{2k} \approx 3 \ \ \ \mbox{for
every} \ \ 1\le k\le n. 
\end{equation}
Next for a given two-player game, additional buyers are inserted in the price-regulating market and game parameters are embedded into their budget and utility functions. These new buyers are given very small budget so that the price-regulation property is still satisfied. 

First, we modify the family of price-regulating markets $\M_n$ for each $n\ge 1$ so that each PLC function is either linear or linear with a threshold. In the construction of~\cite{ChenT09}, each buyer $k$ derives non-zero utility only from goods $2k-1$ and $2k$. Its utility function for good $2k-1$ is linear with slope 2 (utility per unit amount), and for good $2k$ it is linear with slope 4 till unit amount and then linear with slope $1$. Since the slope of the second segment is $1$, it is not good-based budget-additive utility function. Simply decreasing the slope of the second segment from $1$ to $0$ does not work. We get only one inequality: 
\[\frac{1}{2} \le \frac{p_{2k-1}}{p_{2k}}.\]
To construct a correct reduction, we use two buyers, say $(k,1)$ and $(k,2)$, instead of one buyer $k$. We set the supply of each good to $2$ units instead of $1$. Both buyers $(k,1)$ and $(k,2)$ have budget of $3$ units each, and both derive non-zero utility only from goods $2k-1$ and $2k$. We set the utility function of buyer $(k,1)$ as follows: For good $2k-1$, it is linear with slope $2$, and for good $2k$, it is linear with slope $4$ till unit amount and then linear with zero slope. Similarly, the utility function of buyer $(k,2)$ is set as follows: For good $2k$, it is linear with slope $2$, and for good $2k-1$, it is linear with slope $4$ till unit amount and then linear with zero slope. We claim that this enforces the price-regulation property (\ref{eqn:prp}) on every equilibrium price vector $\vecp$. 

Suppose $p_{2k-1}/p_{2k} > 2$ then buyer $(k,2)$ demands only good $2k$. This results in more demand of good $2k$ and less demand of good $2k-1$, hence does not give an equilibrium. Similarly, we get contradiction for the case $p_{2k-1}/p_{2k} < 1/2$. When $\frac{1}{2} \le \frac{p_{2k-1}}{p_{2k}} \le 2$, then buyer $(k,1)$ demands one unit of good $2k-1$ and one unit of good $2k$, and the same for buyer $2k$. This yields an equilibrium. 

Next, for the additional buyers who embed the game parameters, we simply change the slope of the second segment from positive to zero for each utility function. We claim that this works because these buyers do not buy any good on the second segment in the original construction of~\cite{ChenT09}. Hence, it has no effect on equilibrium when the slope of the second segment is decreased. This concludes the proof. 
\end{proof}

\bibliography{../../../Bibfiles/literature,../../../Bibfiles/martin}
\end{document}